\documentclass[runningheads,a4paper]{llncs}

\usepackage{graphicx}
\graphicspath{ {./Figs/} }

\def\a{{\tt a}}
\def\b{{\tt b}}

\spnewtheorem{exa}{Example}{\it}{\rm}

\spnewtheorem*{constraints*}{Constraints}{\bfseries}{\rmfamily}
\spnewtheorem{definition}{Definition}{\bfseries}{\rmfamily}

\usepackage{multicol}
\usepackage{latexsym,amsmath,amssymb,stmaryrd,mathtools}
\usepackage{graphicx,epsfig,tikz}
\usetikzlibrary{shapes.geometric,positioning,chains,fit,calc}
\usetikzlibrary{patterns}
\usepackage{makecell, longtable, lscape, ctable,lscape, supertabular}
\usepackage{csvsimple}
\usepackage{supertabular}
\usepackage{longtable}

\usepackage[ddmmyyyy]{datetime}
\usepackage[ruled,noend,noline]{algorithm2e}
\usepackage[nocompress]{cite}
\usepackage{enumitem}

\newcommand{\runs}{\textrm{runs}}
\newcommand{\bwt}{\textrm{bwt}}
\newcommand{\rev}{\textrm{rev}}
\newcommand{\Oh}{{\cal O}}
\newcommand{\conj}{\textrm{conj}}

\title{Novel Results on the Number of Runs of the Burrows-Wheeler-Transform}

\author{Sara Giuliani\inst1 \and Shunsuke Inenaga\inst2 \and Zsuzsanna Lipt\'ak\inst1 \and Nicola Prezza\inst3 \and Marinella Sciortino\inst4 \and Anna Toffanello\inst1}

\institute{
Dipartimento di Informatica, University of Verona, Italy, \\
\email{$\{$sara.giuliani$\_$01, zsuzsanna.liptak$\}$@univr.it, anna.toffanello@studenti.univr.it} 
\and 
Department of Informatics, Kyushu University, Fukuoka, Japan, \\
PRESTO, Japan Science and Technology Agency, Kawaguchi, Japan
\email{inenaga@inf.kyushu-u.ac.jp}
\and 
Department of Business and Management, LUISS University, Rome, Italy, 
\email{nprezza@luiss.it}
\and 
Dipartimento di Matematica e Informatica, University of Palermo, Italy, 
\email{marinella.sciortino@unipa.it}
}

\authorrunning{S.\ Giuliani, S.\ Inenaga, Zs.\ Lipt\'ak, N.\ Prezza, M.\ Sciortino, A.\ Toffanello}
\titlerunning{}

\begin{document}

\maketitle

\begin{abstract}
The Burrows-Wheeler-Transform (BWT), a reversible string transformation, is one of the fundamental components of many current data structures in string processing.  It is central in data compression, as well as in efficient query algorithms for sequence data, such as webpages, genomic and other biological sequences, or indeed any textual data.  The BWT lends itself well to compression because its number of equal-letter-runs (usually referred to as $r$) is often considerably lower than that of the original string; in particular, it is well suited for strings with many repeated factors.  In fact, much attention has been paid to the $r$ parameter as measure of repetitiveness, especially to evaluate the performance in terms of both space and time of compressed indexing data structures. 

In this paper, we investigate $\rho(v)$, the ratio of $r$ and of the number of runs of the BWT of the reverse of $v$.  Kempa and Kociumaka [FOCS 2020] gave the first non-trivial upper bound as $\rho(v) = O(\log^2(n))$, for any string $v$ of length $n$.  However, nothing is known about the tightness of this upper bound.  We present infinite families of binary strings for which $\rho(v) = \Theta(\log n)$ holds, thus giving the first non-trivial lower bound on $\rho(n)$, the maximum over all strings of length $n$.  

Our results suggest that $r$ is not an ideal measure of the repetitiveness of the string, since the number of repeated factors is invariant between the string and its reverse.  We believe that there is a more intricate relationship between the number of runs of the BWT and the string's combinatorial properties. 
\end{abstract}

\begin{keywords} Burrows-Wheeler-Transform, 
compressed data structures, 
string indexing, 
repetitiveness, 
combinatorics on words
\end{keywords}


\section{Introduction}\label{sec:intro}

Since its introduction in 1994 by Michael Burrows and David J. Wheeler, the Burrows-Wheeler Transform (BWT) \cite{BurrowsWheeler94} has played a fundamental role in lossless data compression and string-processing algorithms. The BWT of a word $w$ can be obtained by concatenating the last characters of the lexicographically-sorted conjugates (that is, rotations) of $w$. Among its many fundamental properties, this permutation turns out to be invertible and more compressible than the original word $w$. The latter property follows from the fact that sorting the conjugates of $w$ has the effect of clustering together repeated factors; as a consequence, characters preceding those repetitions are clustered together in the BWT, and thus repetitions in $w$ tend to generate long runs of equal characters in its BWT. The more repetitive $w$, the lower the number $r$ of such runs. 
This fact motivated recent research on data structures whose size is bounded as a function of $r$: the most prominent example in this direction, the \emph{r-index} \cite{GNP20}, is a fully-compressed index of size $\Oh(r)$ able to locate factor occurrences in log-logarithmic time each. Other examples of recent algorithms working in runs-bounded space include index construction \cite{Kempa19} and data compression in small working space \cite{bannai2018online,PP17,ohno2018faster}.

As it turns out, $r$ is a member of a much larger family of word-repetitiveness measures that have lately generated much interest in the research community. Examples of those measures include (but are not limited to) the number $z$ of factors in the LZ77 factorization \cite{LZ76}, the number $g$ of rules in the smallest context-free grammar generating the word \cite{KY00}, the size $b$ of the smallest bidirectional macro scheme
\cite{SS82}, 
and the size $e$ of the CDAWG \cite{blumer1987complete}. More recently, it was shown that all those compressors are particular cases of a combinatorial object  named \emph{string attractor} \cite{KP18} whose size $\gamma$ lower-bounds all measures $r$, $z$, $g$, $b$, and $e$. In turn, in \cite{KNPlatin20} it was shown that $\gamma$ is lower-bounded by another measure, $\delta$, which is linked to factor complexity (that is, to the number of distinct factors of each length) and better captures the word's repetitiveness. On the upper-bound side, the papers \cite{KNPlatin20,KP18} provided approximation ratios of all measures but $r$ with respect to $\gamma$. Finding an upper-bound for $r$ remained an open problem until the recent work of Kempa and Kociumaka \cite{KK19}, who showed that, for any word of length $n$, $r=\Oh(\delta \log^2 n)$ (which in turn implies $r = \Oh(\gamma \log^2 n)$). As stated explicitly in~\cite{KK19}, this implies the first upper bound on the ratio $\rho$ between $r$ and the number of runs in the BWT of the reverse of the word, namely $\rho = \Oh(\log^2 n)$.

This leaves open the interesting question of whether this bound is tight. In this paper, we give a first answer to this question by exhibiting an infinite  family of binary words whose members satisfy $\rho = \Theta(\log n)$. This contrasts the experimental observation made in \cite{BCGPR15,PP17} that $\rho$ appears to be constant on real repetitive text collections, and shows that $r$ is not a strong repetitiveness measure since---unlike $b$, $g$, $\gamma$, and $\delta$---it is not invariant under reversal.

An added value of the proof we present lies in a surprising insight into the exact structure of the BWT matrix of the words we study: right-extensions of Fibonacci words. This insight allows us to further extend the method to a much larger family of words, giving the number of runs of the BWT for both the word and its reverse, for right-extensions of all standard words. 
As it turns out, the words we obtain from Fibonacci words are maximal with respect to $\rho$ within this class. At the same time, we have verified experimentally that these words are not maximal among all words of the same length. This leaves a gap on the maximum on $\rho$, taken over all words of length $n$,  between our lower bound $\Omega(\log n)$ and the upper bound of $\Oh(\log^2 n)$ of~\cite{KK19}. 

As a matter of fact, the reverse of the Fibonacci extensions allow us to prove an even more surprising result: a single character extension can increase $r$ by a \emph{multiplicative} factor $\Theta(\log n)$. This result is the equivalent of the ``one-bit catastrophe'' exhibited by Lagarde and Perifel \cite{Lagarde18} for Lempel-Ziv '78: using these compression schemes, the compression ratio of the word can change dramatically if just one bit is prepended to the input.


\section{Basics}\label{sec:basics}

Let $\Sigma = \{a,b\}$, with $a<b$. A {\em binary word} (or {\em string}) $w$ is a finite sequence of {\em characters} (or {\em letters}) from $\Sigma$. We denote the $i$th character of $w$ by $w[i]$ and index words from $1$. We denote by $|w|$ the length of $w$, and by $|w|_a$ resp.\ $|w|_b$ the number of characters $a$ resp.\ $b$ in $w$. The empty word $\epsilon$ is the unique word of length $0$. The set of words over $\Sigma$ is denoted $\Sigma^*$. We write $w^{\rev} = w[n]\cdots w[1]$ for the {\em reverse}
of a word $w$ of length $n$.
The word $w'$ is a {\em conjugate} of the word $w$ if $w'= w[i]\cdots w[n]w[1]\cdots w[i-1] =: \conj_i(w)$ for some $i=1,\ldots, n$ (also called the {\em $i$th  rotation} of $w$). 

If $w=uxv$, for some words $u,x,v \in \Sigma^*$, then $u$ is called a {\em prefix}, $v$ a {\em suffix}, and $x$ a {\em factor} of $w$. A prefix (suffix, factor) $u$ of $w$ is called {\em proper}  if $u\neq w$. A word $u$ is a {\em circular factor} of $w$ if it is the prefix of some conjugate of $w$. A circular factor $u$ is called {\em left-special} if both $au$ and $bu$ occur as circular factors. For an integer $k\geq 1$, $u^k = u\cdots u$ is the $k$th power of $u$. A word $w$ is called {\em primitive} if $w = u^k$ implies $k=1$. A word $w$ is primitive if and only if it has exactly $|w|$ distinct conjugates. 

For two words $v,w$, the {\em longest common prefix} $lcp(v,w)$ is defined as the maximum length word $u$ such that $u$ is a prefix both of $v$ and of $w$. The {\em lexicographic order} on $\Sigma^*$ is defined by: $v < w$ if either $v$ is a proper prefix of $w$, or $ua$ is a prefix of $v$ and $ub$ is a prefix of $w$, where $u = lcp(v,w)$. A {\em Lyndon word} is a primitive word which is lexicographically smaller than all of its conjugates.  To simplify the discussion, we will assume from now on that $w$ is primitive (but everything can be extended also to non-primitive words). 

The {\em Burrows-Wheeler-Transform} (BWT)~\cite{BurrowsWheeler94} of a word $w$ of length $n$ is a permutation of the characters of $w$, defined as the sequence of final characters of the lexicographically ordered set of conjugates of $w$. More precisely, let the {\em BW-array} be an array of size $n$ defined as: $BW[i] = k$ if $\conj_k(w)$ is the $i$th conjugate of $w$ in lexicographic order.\footnote{Note that this is in general not the same as the suffix array SA, since here we have the conjugates and not the suffixes.} Then $\bwt(w)[i] = w[{BW[i]-1}]$, where we set $w[0] = w[n]$. Another way to visualize the BWT is via an $(n\times n)$-matrix containing the lexicographically sorted conjugates of $w$: the BWT of $w$ equals the last column of this matrix, read from top to bottom, see Fig.~\ref{fig:example_s8b}. By definition, $\bwt(w) = \bwt(w')$ if and only if $w$ and $w'$ are conjugates. 

For a word $w$, let $\runs(w)$ denote the number of maximal equal-letter runs of $w$, and $r(w) = \runs(\bwt(w))$. We are now ready for our main definition: 

\begin{definition} Let $w\in \{a,b\}^*$. We define the {\em runs-ratio} $\rho(w)$ as 
\begin{align*}
\rho(w) &= \max \left(\frac{\runs(\bwt(w))}{\runs(\bwt(w^{\rev}))}, 
\frac{\runs(\bwt(w^{\rev}))}{\runs(\bwt(w))}\right) = \max \left( \frac{r(w)}{r(w^{\rev})}, \frac{r(w^{\rev})}{r(w)} \right), 
\end{align*}

\noindent and $\rho(n) = \max \{ \rho(w) : |w| = n\}$. 
\end{definition}

Note that $\rho(w) \geq 1$ holds by definition.  Since $r(w) = r(w^{\rev})$ for all $w$ with $|w|\leq 6$, we have $\rho(n) =1$ for $n<7$.  In Table~\ref{tab:rhon1}, we give the values of $\rho(n)$ for $n=7, \ldots, 30$ (computed with a computer program): 

\begin{table}\centering
\begin{tabular}{c|*{24}{c}}
$n$ & 7 &  8 & 9 & 10 & 11 & 12 & 13 & 14 & 15 & 16 & 17 & 18 & 19 & 20 & 21 & 22 & 23 & 24 & 25 & 26 & 27 & 28 & 29 & 30\\
\hline
$\rho(n)$ & 1.5 & 1.5 & 2 & 2 & 2 & 2 & 2 & 2 & 2 &  2 & 2.5 & 2.5 & 2.5 & 2.5 & 3 & 2.5 & 3 & 3 & 2.67 & 3 & 3 & 3 & 3 & 3\\
\end{tabular}
\caption{The values of $\rho(n)$ for $n=7,\ldots, 30$. \label{tab:rhon1}}
\end{table}

We introduce {\em standard words} next, following~\cite{deLuca97a}. Given an infinite sequence of integers $(d_0,d_1,d_2, \ldots)$, with $d_0\geq 0, d_i>0$ for all $i>0$, called a {\em directive sequence}, define a sequence of words $(s_{i})_{i\geq 0}$ of increasing length as follows: $s_0 = b, s_1 = a, s_{i+1} = s_i^{d_{i-1}}s_{i-1}$, for $i\geq 1$. The index $i$ is referred to as the {\em order} of $s_i$.  
The best known example is the sequence of {\em Fibonacci words}, which are given by the directive sequence $(1,1,1,\ldots)$, and of which the first few elements are as follows: 
\begin{align*} 
s_0 &= b, s_1 =a, s_2=ab, s_3=aba, s_4=abaab, s_5 = abaababa, s_6 = abaababaabaab,\\ 
s_7 &= abaababaabaababaababa, s_8 = abaababaabaababaababaabaababaabaab, \ldots
\end{align*}

Note that $|s_i| = F_i$, where $F_i$ is the Fibonacci sequence, defined by $F_0=F_1=1$ and $F_{i+1} = F_i + F_{i-1}$. Moreover, $|s_i|_a = F_{i-1}$ and $|s_i|_b = F_{i-2}$, for $i\geq 2$. 

Standard words are used for the construction of infinite Sturmian words, in the  sense  that  every  characteristic Sturmian word is the limit of a sequence of standard words (cf. Chapter 2 of \cite{Loth2}). 
These words have many interesting combinatorial properties and appear as extreme case in a great range of contexts \cite{KnuthMorrisPratt1977, deLucaMignosi1994, deLuca97, CastiglioneRS08 ,CS09}. 
 A fundamental result in connection with the BWT is the following: $\bwt(w) = b^qa^p$ with $\gcd(q,p)=1$ if and only if $w$ is a standard word~\cite{MantaciRS03}.


\section{Fibonacci-plus words have $\rho = \Theta(\log n)$} 

Since for a standard word $s$, $s^{\rev}$ is a conjugate, we have $\rho(s) = 1$  for all standard words $s$. 
We will show in this section that adding just one character at the end of the word suffices to increase $\rho$ from $1$ to logarithmic in the length of the word.  

\begin{definition}
A word $v$ is called {\em Fibonacci-plus} if it is either of the form $sb$, where $s$ is a Fibonacci word of even order $2k$, $k\geq 2$, or of the form $sa$, where $s$ is a Fibonacci word of odd order $2k+1$,  $k\geq 2$. In the first case, $v$ is {\em of even order}, otherwise {\em of odd order}.  
\end{definition}

The aim of this section is to prove the following theorem: 

\begin{theorem}\label{thm:fibplus}
Let $v$ be a Fibonacci-plus word, and let $|v|=n$. Then $\rho(s) = \Theta(\log n)$. 
\end{theorem}

We will prove the theorem by showing that, for a Fibonacci-plus word $v$, $r(v) = 4$ (Prop.~\ref{prop:fibplus-r}) and $r(v^{\rev})$ is linear in the order of the word itself (Prop.~\ref{prop:fibplusrev-r-even}). The statement will then follow by an argument on the length of $v$. 

Fibonacci words have very well-known structural and combinatorial properties \cite{deLuca81}, some of them can be deduced from more general properties that hold true for all standard words (see \cite{deLucaMignosi1994,deLuca97a,BersteldeLuca97,BorelR06}). In the next proposition we summarize some of these properties, which will be useful in the following.

\begin{proposition}[Some known properties of the Fibonacci words] \label{prop:fib}
Let $s_i$ be the Fibonacci word of order $i\geq 0$. The following properties hold:
\begin{enumerate}
    \item for all $k\geq 1$, $s_{2k}=x_{2k}ab$ and $s_{2k+1}=x_{2k+1}ba$, where $x_{2k}$ and $x_{2k+1}$ are  palindromes ($x_2=\epsilon$). 
    \item for all $k\geq 2$,
    \begin{itemize}
        \item $s_{2k}=x_{2k-1}bax_{2k-2}ab=x_{2k-2}abx_{2k-1}ab$
        \item $s_{2k+1}=x_{2k}abx_{2k-1}ba=x_{2k-1}bax_{2k}ba$.
    \end{itemize}  
   \item for all $i\geq 2$, $a x_i b$ is a Lyndon word.
   \item for all circular factors $y,z$ of $s_i$ with $|y| = |z|$, and for each $c \in \Sigma$, 
   one has that $||y|_c - |z|_c| \leq  1$ (Balancedness Property). 
\end{enumerate}
\end{proposition}

\begin{example}
Let us consider $s_8=abaababaabaababaababaabaababaabaab$ the Fibonacci word of order $8$ and length $F_8=34$. 

One can verify that the prefix  $x_8=abaababaabaababaababaabaababaaba$ is a palindrome. Moreover $x_8=x_7 ba x_6=x_6 ab x_7$, where $x_7=abaababaabaababaaba$ and $x_6=abaababaaba$.
\end{example}

\begin{proposition}\label{prop:fibplus-r} Let $v$ be a Fibonacci-plus word. Then $r(v) = 4$. In particular, 
\begin{enumerate}
\item if $v = s_{2k}b$, then $\bwt(v) = b^{F_{2k-2}}a^{F_{2k-1}-1}ba$, and 
\item if $v=s_{2k+1}a$, then $\bwt(v) = 
bab^{F_{2k-1}-1}a^{F_{2k}}$.  
\end{enumerate}

\end{proposition}

\begin{proof} We give the proof for even order only. The proof for odd order is analogous. 

Let us write $v= sb$, with $s=s_{2k}$. Since Fibonacci words are standard words, it follows that $\bwt(s) = b^{F_{2k-2}}a^{F_{2k-1}}$ (see Sec.~\ref{sec:basics}). Since $s$ is of even order, it can be written as $s = xab$ for a palindrome $x$ (Prop.~\ref{prop:fib}, part 1); moreover, it follows  from the specific form of $x$ (Prop.~\ref{prop:fib}, part 2) that both $xab$ and $xba$ are conjugates. It is further known that the two conjugates $xab$ and $xba$ are at position $F_{2k-2}$ and $F_{2k-2}+1$, respectively, i.e.\ they correspond to the last $b$ and the first $a$ in the BWT of $s$~\cite{MantaciRS03,BorelR06}.  From this it follows that, for all $h$,
\begin{equation} \label{eq:lasta}
v[h] = a, v[h-1] = b \quad \Rightarrow \quad \conj_h(s) < xab.
\end{equation}

Now consider the conjugates of $v=sb$. Clearly, $\conj_{n-1}(v)$ is the largest, since it starts with $bb$, and $s$ contains no factor $bb$. 

Next we show that the penultimate conjugate is $\conj_{n}(v) = bxab$. In order to prove this, we need to show that $\conj_i(v) < bxab$ for all $i<n-1$. If $\conj_i(v)$ begins with $a$, then this is clearly true. Otherwise, $\conj_i(v) = bv[i+1]\cdots v[i-1]$. Since $v[i]=b$, we have that  $\conj_{i+1}(v) < xab$ by Eq.~\eqref{eq:lasta}, and therefore, $\conj_i(v) < bxab$, as claimed. 

We have so far explained the last two characters of the BWT. Now we will show that the remaining part of the $BW$-arrays coincides for the two words $s$ and $v$.  A quick inspection shows that the last character of $\conj_i(s)$ and $\conj_i(v)$ is the same, for all $i<n-1$, which yields the desired form of the BWT of $v$. 

We will prove that the relative order of two consecutive conjugates of $s$ is the same in $v$, i.e.\ that the insertion of the $b$ at the end of $s$ does not change this order. This will imply that the $BW$-array remains the same for the first $n-2$ entries. 

Let $\conj_i(s)<\conj_j(s)$ be consecutive conjugates of $s$. If $i<j$, then the new $b$ appears earlier in $\conj_j(s)$ than in $\conj_i(s)$, therefore $\conj_i(v)<\conj_j(v)$ clearly holds. Now let $i>j$. It is known~\cite{BorelR06} that two consecutive conjugates of $s$ have the form $uabu'$ and $ubau'$, where $u'u = x$ is the palindrome from Prop.~\ref{prop:fib}, part 2. From $s_{2k}=x_{2k-1}bax_{2k-2}ab=x_{2k-2}abx_{2k-1}ab$, it follows that $x_{2k} = x_{2k-1}bax_{2k-2} = x_{2k-2}abx_{2k-1}$, and we deduce that $x=x_{2k}$ has exactly two occurrences in $s$ as a circular factor. Therefore, $\conj_i(v) = uabbu'$ and the new $b$ appears in $\conj_j(v)$ within the suffix $u'$. This implies $u = lcp(\conj_i(v),\conj_j(v))$, and thus $\conj_i(v) < \conj_j(v)$. 

This completes the proof.
\end{proof}

The next proposition gives the form of the BWT of the reverse. 

\begin{proposition}\label{prop:fibplusrev-r-even} \label{prop:fibplusrev}
Let $v$ be a Fibonacci-plus word. Then $r(v^{\rev}) = 2k$. In particular, 
\begin{enumerate}
\item if $v$ is of even order, i.e.\ $v=s_{2k}b$ for some $k\geq 1$, then 
$\bwt(v^{\rev}) = b^{F_{2k-2} - k + 1} a^{F_0} b a^{F_2} b a^{F_4} b \cdots a^{F_{2k-4}}bba^{F_{2k-2}}$, 
\item if $v$ is of odd order, i.e.\ $v=s_{2k+1}a$ for some $k\geq 1$, then $\bwt(v^{\rev}) = b^{F_{2k-2}}aab^{F_{2k-4}}ab^{F_{2k-6}}a\cdots b^{F_2}ab^{F_0}a^{F_{2k}-k+1}$. 
\end{enumerate}
\end{proposition}

\begin{example}
In Figure~\ref{fig:example_s8b} we display the BWT-matrices of the Fibonacci-plus word $v =s_{8}b$ of length $35$ and of its reverse.
\end{example}

\begin{figure}[h!]\hspace*{-0.6in}\scalebox{.9}{
\begin{minipage}{0.68\textwidth}
	\begin{tabular}{r r c@{\hskip.01pt}l@{}}		
		\multicolumn{1}{p{.1cm}}{\centering BW \\ array}&   & \multicolumn{1}{p{6.3cm}}{\centering rotations of $v=$ \\ {\tt  {\footnotesize abaababaabaababaababaabaababaabaab{\underline b}}}} & \multicolumn{1}{p{.01cm}}{ $\bwt(v)$} \\
		\hline
		1 & 21 & {\tt aabaababaabaab{\underline b}abaababaabaababaabab}  & \b \\ 
		2 & 8 & {\tt aabaababaababaabaababaabaab{\underline b}abaabab}  & \b \\ 
		3 & 29 & {\tt aabaab{\underline b}abaababaabaababaababaabaabab}  & \b \\ 
		4 & 16 & {\tt aababaabaababaabaab{\underline b}abaababaabaabab}  & \b \\ 
		5 & 3 & {\tt aababaabaababaababaabaababaabaab{\underline b}ab}  & \b \\ 
		6 & 24 & {\tt aababaabaab{\underline b}abaababaabaababaababaab}  & \b \\ 
		7 & 11 & {\tt aababaababaabaababaabaab{\underline b}abaababaab}  & \b \\ 
		8 & 32 & {\tt aab{\underline b}abaababaabaababaababaabaababaab}  & \b \\ 
		9 & 19 & {\tt abaabaababaabaab{\underline b}abaababaabaababaab}  & \b \\ 
		10 & 6 & {\tt abaabaababaababaabaababaabaab{\underline b}abaab}  & \b \\ 
		11 & 27 & {\tt abaabaab{\underline b}abaababaabaababaababaabaab}  & \b \\ 
		12 & 14 & {\tt abaababaabaababaabaab{\underline b}abaababaabaab}  & \b \\ 
		13 & 1 & {\tt abaababaabaababaababaabaababaabaab{\underline b}}  & {\underline \b} \\ 
		14 & 22 & {\tt abaababaabaab{\underline b}abaababaabaababaababa}  & \a \\ 
		15 & 9 & {\tt abaababaababaabaababaabaab{\underline b}abaababa}  & \a \\ 
		16 & 30 & {\tt abaab{\underline b}abaababaabaababaababaabaababa}  & \a \\ 
		17 & 17 & {\tt ababaabaababaabaab{\underline b}abaababaabaababa}  & \a \\ 
		18 & 4 & {\tt ababaabaababaababaabaababaabaab{\underline b}aba}  & \a \\ 
		19 & 25 & {\tt ababaabaab{\underline b}abaababaabaababaababaaba}  & \a \\ 
		20 & 12 & {\tt ababaababaabaababaabaab{\underline b}abaababaaba}  & \a \\ 
		21 & 33 & {\tt ab{\underline b}abaababaabaababaababaabaababaaba}  & \a \\ 
		22 & 20 & {\tt baabaababaabaab{\underline b}abaababaabaababaaba}  & \a \\ 
		23 & 7 & {\tt baabaababaababaabaababaabaab{\underline b}abaaba}  & \a \\ 
		24 & 28 & {\tt baabaab{\underline b}abaababaabaababaababaabaaba}  & \a \\ 
		25 & 15 & {\tt baababaabaababaabaab{\underline b}abaababaabaaba}  & \a \\ 
		26 & 2 & {\tt baababaabaababaababaabaababaabaab{\underline b}a}  & \a \\ 
		27 & 23 & {\tt baababaabaab{\underline b}abaababaabaababaababaa}  & \a \\ 
		28 & 10 & {\tt baababaababaabaababaabaab{\underline b}abaababaa}  & \a \\ 
		29 & 31 & {\tt baab{\underline b}abaababaabaababaababaabaababaa}  & \a \\ 
		30 & 18 & {\tt babaabaababaabaab{\underline b}abaababaabaababaa}  & \a \\ 
		31 & 5 & {\tt babaabaababaababaabaababaabaab{\underline b}abaa}  & \a \\ 
		32 & 26 & {\tt babaabaab{\underline b}abaababaabaababaababaabaa}  & \a \\ 
		33 & 13 & {\tt babaababaabaababaabaab{\underline b}abaababaabaa}  & \a \\ 
		34 & 35 & {\tt {\underline b}abaababaabaababaababaabaababaabaab}  & \b \\ 
		35 & 34 & {\tt b{\underline b}abaababaabaababaababaabaababaabaa}  & \a \\
	\end{tabular}			
\end{minipage}%
\vline
\begin{minipage}{1\textwidth}
	\begin{tabular}{r r c l@{}}			
		\multicolumn{1}{p{.1cm}}{\centering BW \\ array}&   & \multicolumn{1}{p{6.3cm}}{\centering rotations of $v^{rev}=$ \\ {\tt  {\footnotesize {\underline b}baabaababaabaababaababaabaababaaba}}} & \multicolumn{1}{p{.1cm}}{ $\bwt(v^{rev})$} \\
		\hline
		1 & 3 & {\tt aabaababaabaababaababaabaababaaba{\underline b}b}  & \b \\ 
		2 & 11 & {\tt aabaababaababaabaababaaba{\underline b}baabaabab}  & \b \\ 
		3 & 24 & {\tt aabaababaaba{\underline b}baabaababaabaababaabab}  & \b \\ 
		4 & 6 & {\tt aababaabaababaababaabaababaaba{\underline b}baab}  & \b \\ 
		5 & 19 & {\tt aababaabaababaaba{\underline b}baabaababaabaabab}  & \b \\ 
		6 & 14 & {\tt aababaababaabaababaaba{\underline b}baabaababaab}  & \b \\ 
		7 & 27 & {\tt aababaaba{\underline b}baabaababaabaababaababaab}  & \b \\ 
		8 & 32 & {\tt aaba{\underline b}baabaababaabaababaababaabaabab}  & \b \\ 
		9 & 9 & {\tt abaabaababaababaabaababaaba{\underline b}baabaab}  & \b \\ 
		10 & 22 & {\tt abaabaababaaba{\underline b}baabaababaabaababaab}  & \b \\ 
		11 & 4 & {\tt abaababaabaababaababaabaababaaba{\underline b}ba}  & \a \\ 
		12 & 17 & {\tt abaababaabaababaaba{\underline b}baabaababaabaab}  & \b \\ 
		13 & 12 & {\tt abaababaababaabaababaaba{\underline b}baabaababa}  & \a \\ 
		14 & 25 & {\tt abaababaaba{\underline b}baabaababaabaababaababa}  & \a \\ 
		15 & 30 & {\tt abaaba{\underline b}baabaababaabaababaababaabaab}  & \b \\ 
		16 & 7 & {\tt ababaabaababaababaabaababaaba{\underline b}baaba}  & \a \\ 
		17 & 20 & {\tt ababaabaababaaba{\underline b}baabaababaabaababa}  & \a \\ 
		18 & 15 & {\tt ababaababaabaababaaba{\underline b}baabaababaaba}  & \a \\ 
		19 & 28 & {\tt ababaaba{\underline b}baabaababaabaababaababaaba}  & \a \\ 
		20 & 33 & {\tt aba{\underline b}baabaababaabaababaababaabaababa}  & \a \\ 
		21 & 35 & {\tt a{\underline b}baabaababaabaababaababaabaababaab}  & \b \\ 
		22 & 2 & {\tt baabaababaabaababaababaabaababaaba{\underline b}}  & {\underline \b} \\ 
		23 & 10 & {\tt baabaababaababaabaababaaba{\underline b}baabaaba}  & \a \\ 
		24 & 23 & {\tt baabaababaaba{\underline b}baabaababaabaababaaba}  & \a \\ 
		25 & 5 & {\tt baababaabaababaababaabaababaaba{\underline b}baa}  & \a \\ 
		26 & 18 & {\tt baababaabaababaaba{\underline b}baabaababaabaaba}  & \a \\ 
		27 & 13 & {\tt baababaababaabaababaaba{\underline b}baabaababaa}  & \a \\ 
		28 & 26 & {\tt baababaaba{\underline b}baabaababaabaababaababaa}  & \a \\ 
		29 & 31 & {\tt baaba{\underline b}baabaababaabaababaababaabaaba}  & \a \\ 
		30 & 8 & {\tt babaabaababaababaabaababaaba{\underline b}baabaa}  & \a \\ 
		31 & 21 & {\tt babaabaababaaba{\underline b}baabaababaabaababaa}  & \a \\ 
		32 & 16 & {\tt babaababaabaababaaba{\underline b}baabaababaabaa}  & \a \\ 
		33 & 29 & {\tt babaaba{\underline b}baabaababaabaababaababaabaa}  & \a \\ 
		34 & 34 & {\tt ba{\underline b}baabaababaabaababaababaabaababaa}  & \a \\ 
		35 & 1 & {\tt {\underline b}baabaababaabaababaababaabaababaaba}  & \a \\
	\end{tabular}
\end{minipage}
}\caption{BWT-matrices of the Fibonacci-plus word $v =s_{8}b$ of length $35$ and its reverse, underlined the added $b$.\label{fig:example_s8b}}
\end{figure}

The rest of this section is devoted to the proof of Proposition~\ref{prop:fibplusrev-r-even}. We will prove the case of even order $v$ only; an analogous argument proves the case of odd order $v$. Our proof is based on a detailed analysis of the structure of the BWT matrix of $v^{\rev}$. We will divide the BWT-matrix, and thus the BWT, into three parts, based on the positions of three specific conjugates of $v^{\rev}$, and analyse each of these separately.

Now consider the first few conjugates of $v^{\rev}$. Since $v = s_{2k}b = x_{2k}abb$, we have $v^{\rev} = bbax_{2k}$, noting that $x_{2k}$ is a palindrome. Thus 
\begin{align*}
\conj_1(v^{\rev}) & = bbax_{2k}, \\
\conj_2(v^{\rev}) & = bax_{2k}b, \\
\conj_3(v^{\rev}) & = ax_{2k}bb, \\
\conj_4(v^{\rev}) & = x_{2k}bba.
\end{align*}

Since Fibonacci words have no occurrence of $bb$, the conjugate $\conj_1(v^{\rev}) = v^{\rev}$ is the last row of the matrix. Moreover, by Prop.~\ref{prop:fib}, $ax_{2k}b$ is a Lyndon word, and therefore $\conj_3(v^{\rev})$, having only an extra $b$ at the end, is also Lyndon, and thus can be found in the first row. The relative order of the other two conjugates is also clear, since $x_{2k}$ begins with an $a$, thus we have 

\[ ax_{2k}bb < x_{2k}bba < bax_{2k}b < bbax_{2k}.\]

We will now subdivide the BWT-matrix into three parts, according to the positions of these conjugates, and we will call these \emph{top part}, \emph{middle part}, and \emph{bottom part}. The conjugates $ax_{2k}bb$, $x_{2k}bba$ and $bax_{2k}b$ are the first row of the top part, middle part and bottom part, respectively.  We use this to partition the BWT into the three corresponding parts $\bwt(v^{\rev})_{{\rm top}}, \bwt(v^{\rev})_{{\rm mid}},$ and $\bwt(v^{\rev})_{{\rm bot}}$. 
Thus we have 
\[\bwt(v^{\rev}) = \bwt(v^{\rev})_{{\rm top}}\cdot \bwt(v^{\rev})_{{\rm mid}}\cdot \bwt(v^{\rev})_{{\rm bot}}.\]

We will prove the form of the BWT of $v^{\rev}$ separately for the three parts. In Fig.~\ref{fig:stdplus-matrix} we give a visual presentation of the proof. 

	\begin{figure}\centering

\tikzset{
	pattern size/.store in=\mcSize, 
	pattern size = 5pt,
	pattern thickness/.store in=\mcThickness, 
	pattern thickness = 0.3pt,
	pattern radius/.store in=\mcRadius, 
	pattern radius = 1pt}
\makeatletter
\pgfutil@ifundefined{pgf@pattern@name@_su6fmrwbi}{
	\pgfdeclarepatternformonly[\mcThickness,\mcSize]{_su6fmrwbi}
	{\pgfqpoint{0pt}{0pt}}
	{\pgfpoint{\mcSize+\mcThickness}{\mcSize+\mcThickness}}
	{\pgfpoint{\mcSize}{\mcSize}}
	{
		\pgfsetcolor{\tikz@pattern@color}
		\pgfsetlinewidth{\mcThickness}
		\pgfpathmoveto{\pgfqpoint{0pt}{0pt}}
		\pgfpathlineto{\pgfpoint{\mcSize+\mcThickness}{\mcSize+\mcThickness}}
		\pgfusepath{stroke}
}}
\makeatother
\tikzset{every picture/.style={line width=0.75pt}} 
		\hspace*{-1in}\scalebox{1}{
		\begin{tikzpicture}[x=0.75pt,y=0.75pt,yscale=-1,xscale=1]
		

		\draw   (140,11) -- (490,11) -- (490,115.81) -- (140,115.81) -- cycle ;
		\draw   (140,410.36) -- (490,410.36) -- (490,558.5) -- (140,558.5) -- cycle ;
		\draw   (140,115.81) -- (490,115.81) -- (490,134.73) -- (140,134.73) -- cycle ;
		\draw  (140,31.92)  -- (491,31.92);
		
		\draw [pattern=_su6fmrwbi,pattern size=9pt,pattern thickness=0.75pt,pattern radius=0pt] (140,11) -- (290,11) -- (290,31.92) -- (140,31.92) -- cycle;
		\draw [pattern=_su6fmrwbi,pattern size=9pt,pattern thickness=0.75pt,pattern radius=0pt] (343,11) -- (490,11) -- (490,31.92) -- (343,31.92) -- cycle;
		
		\draw [pattern=_su6fmrwbi,pattern size=9pt,pattern thickness=0.75pt,pattern radius=0pt] (140,115.64) -- (290,115.64) -- (290,134.73) -- (140,134.73) -- cycle;
		\draw [pattern=_su6fmrwbi,pattern size=9pt,pattern thickness=0.75pt,pattern radius=0pt] (343,115.64) -- (490,115.64) -- (490,134.73) -- (343,134.73) -- cycle;

		\draw [pattern=_su6fmrwbi,pattern size=9pt,pattern thickness=0.75pt,pattern radius=0pt] (140,410.36) -- (290,410.36) -- (290,432) -- (140,432) -- cycle;
		\draw [pattern=_su6fmrwbi,pattern size=9pt,pattern thickness=0.75pt,pattern radius=0pt] (343,410.36) -- (490,410.36) -- (490,432) -- (343,432) -- cycle;
		\draw (140,432) -- (490,432);
		
		\draw [pattern=_su6fmrwbi,pattern size=9pt,pattern thickness=0.75pt,pattern radius=0pt] (140,535) -- (290,535) -- (290,558.5) -- (140,558.5) -- cycle;
		\draw [pattern=_su6fmrwbi,pattern size=9pt,pattern thickness=0.75pt,pattern radius=0pt] (343,535) -- (490,535) -- (490,558.5) -- (343,558.5) -- cycle;
		\draw (140,535) -- (490,535);


		\draw    (120,62.5) -- (96.5,62.27) ;
		\draw [shift={(94.5,62.25)}, rotate = 360.56] [color={rgb, 255:red, 0; green, 0; blue, 0 }  ][line width=0.75]    (10.93,-3.29) .. controls (6.95,-1.4) and (3.31,-0.3) .. (0,0) .. controls (3.31,0.3) and (6.95,1.4) .. (10.93,3.29)   ;
		\draw    (120.5,10.5) -- (119.5,114.75) ;
		\draw    (119.5,114.75) -- (133,114.75) ;
		\draw    (120.5,10.5) -- (134,10.5) ;
		\draw    (118.5,256.79) -- (95,256.17) ;
		\draw [shift={(93,256.12)}, rotate = 361.5] [color={rgb, 255:red, 0; green, 0; blue, 0 }  ][line width=0.75]    (10.93,-3.29) .. controls (6.95,-1.4) and (3.31,-0.3) .. (0,0) .. controls (3.31,0.3) and (6.95,1.4) .. (10.93,3.29)   ;
		\draw    (119,117.56) -- (118,409.86) ;
		\draw    (118,409.86) -- (131.5,409.86) ;
		\draw    (119,117.5) -- (132.5,117.5) ;
		\draw    (117.5,484.92) -- (94,484.6) ;
		\draw [shift={(92,484.57)}, rotate = 360.78999999999996] [color={rgb, 255:red, 0; green, 0; blue, 0 }  ][line width=0.75]    (10.93,-3.29) .. controls (6.95,-1.4) and (3.31,-0.3) .. (0,0) .. controls (3.31,0.3) and (6.95,1.4) .. (10.93,3.29)   ;
		\draw    (118,412.22) -- (117,557.98) ;
		\draw    (117,557.98) -- (130.5,557.98) ;
		\draw    (118,412.22) -- (131.5,412.22) ;
		\draw    (520.5,62.24) -- (544,62.49) ;
		\draw [shift={(546,62.51)}, rotate = 180.6] [color={rgb, 255:red, 0; green, 0; blue, 0 }  ][line width=0.75]    (10.93,-3.29) .. controls (6.95,-1.4) and (3.31,-0.3) .. (0,0) .. controls (3.31,0.3) and (6.95,1.4) .. (10.93,3.29)   ;
		\draw    (519.96,114.24) -- (521.04,10) ;
		\draw    (519.96,114.24) -- (506.46,114.24) ;
		\draw    (521.04,10) -- (507.54,9.99) ;
		\draw    (519.83,125.72) -- (543.33,125.76) ;
		\draw [shift={(545.33,125.77)}, rotate = 180.1] [color={rgb, 255:red, 0; green, 0; blue, 0 }  ][line width=0.75]    (10.93,-3.29) .. controls (6.95,-1.4) and (3.31,-0.3) .. (0,0) .. controls (3.31,0.3) and (6.95,1.4) .. (10.93,3.29)   ;
		\draw    (519.67,133.67) -- (520.04,116.99) ;
		\draw    (519.67,133.67) -- (506.17,133.67) ;
		\draw    (520.04,116.99) -- (506.54,116.99) ;
		\draw    (520.19,166.58) -- (543.68,166.64) ;
		\draw [shift={(545.68,166.64)}, rotate = 180.15] [color={rgb, 255:red, 0; green, 0; blue, 0 }  ][line width=0.75]    (10.93,-3.29) .. controls (6.95,-1.4) and (3.31,-0.3) .. (0,0) .. controls (3.31,0.3) and (6.95,1.4) .. (10.93,3.29)   ;
		\draw    (520.33,179.67) -- (520.04,153.49) ;
		\draw    (520.33,179.67) -- (506.83,179.66) ;
		\draw    (520.04,153.49) -- (506.54,153.49) ;
		\draw    (521.33,243.34) -- (544.83,243.49) ;
		\draw [shift={(546.83,243.5)}, rotate = 180.36] [color={rgb, 255:red, 0; green, 0; blue, 0 }  ][line width=0.75]    (10.93,-3.29) .. controls (6.95,-1.4) and (3.31,-0.3) .. (0,0) .. controls (3.31,0.3) and (6.95,1.4) .. (10.93,3.29)   ;
		\draw    (520.8,274.15) -- (521,214.07) ;
		\draw    (520.8,274.15) -- (507.3,274.14) ;
		\draw    (521,214.07) -- (507.5,214.06) ;
		\draw    (519,348.59) -- (542.5,348.78) ;
		\draw [shift={(544.5,348.79)}, rotate = 180.46] [color={rgb, 255:red, 0; green, 0; blue, 0 }  ][line width=0.75]    (10.93,-3.29) .. controls (6.95,-1.4) and (3.31,-0.3) .. (0,0) .. controls (3.31,0.3) and (6.95,1.4) .. (10.93,3.29)   ;
		\draw    (520,390.93) -- (519.54,308.58) ;
		\draw    (520,390.93) -- (506.5,390.93) ;
		\draw    (519.54,308.58) -- (506.04,308.57) ;
		\draw    (519,497.28) -- (542.5,497.57) ;
		\draw [shift={(544.5,497.6)}, rotate = 180.71] [color={rgb, 255:red, 0; green, 0; blue, 0 }  ][line width=0.75]    (10.93,-3.29) .. controls (6.95,-1.4) and (3.31,-0.3) .. (0,0) .. controls (3.31,0.3) and (6.95,1.4) .. (10.93,3.29)   ;
		\draw    (519.67,557.71) -- (519.54,435.77) ;
		\draw    (519.67,557.71) -- (506.17,557.7) ;
		\draw    (519.54,435.77) -- (506.04,435.75) ;

		\draw (490,134.79) -- (490,193.25); 
		\draw (490,210.63) -- (490,289.5) ;
		\draw (490,304.5) -- (490, 410.36);
		
		\draw (140,134.79) -- (140,193.25); 
		\draw (140,210.63) -- (140,289.5) ;
		\draw (140,304.5) -- (140, 410.36);

		\draw (140, 193.25) -- (140, 210.63);
		\draw (140, 289.5) -- (140,304.5 );
		
		\draw  (490, 193.25) -- (490, 210.63);
		\draw (490, 289.5) -- (490,304.5 );
		
		\draw (140,193.25) -- (490, 193.25);
		\draw(140,210.63) --(490,210.63);
		\draw (140,289.5) -- (490,289.5);
		\draw (140,304.5)--(490,304.5);
		

		\draw    (140,276.13) -- (239.33,276) ;
		\draw    (239.33,210) -- (239.33,276) ;
		\draw  [dashed]  (139.5,190.88) -- (277.25,190.88) ;
		\draw  [dashed]  (276,178.25) -- (276,190.88) ;
		\draw    (140,175) -- (326.5,175) ;
		\draw    (326.5,156) -- (326.5,175) ;
		\draw (140,154.5) -- (326.5,154.5);
		\draw  [dashed]  (140,406) -- (159,406.25) ;
		\draw  [dashed]  (159,406.25) -- (159,389.25) ;
		\draw  [dashed]  (140,286.75) -- (216.5,286.75) ;
		\draw  [dashed]  (216.5,275.25) -- (216.5,286.75) ;
		\draw  [dashed]  (140,152) -- (388.5,152) ;
		\draw  [dashed]  (388.5,136.25) -- (388.5,152) ;
		\draw    (140,390) -- (187,390) ;
		\draw    (187,305) -- (187,390) ;
		
		\draw (498,119) node [anchor=north west][inner sep=0.75pt]   [align=left] {a};
		\draw (295,120) node [anchor=north west][inner sep=0.75pt]   [align=left] {$x_{2k}$bba};
		\draw (144,138) node [anchor=north west][inner sep=0.75pt]   [align=left] {$x_{2k-1}$};
		\draw (390,138) node [anchor=north west][inner sep=0.75pt]   [align=left] {bb};
		
		\draw (144,159) node [anchor=north west][inner sep=0.75pt]   [align=left] {$x_{2k-2}$b};
		\draw (327,166) node [anchor=north west][inner sep=0.75pt]   [align=left] {b};
		\draw (327,155) node [anchor=north west][inner sep=0.75pt]   [align=left] {a};
		\draw (144,180) node [anchor=north west][inner sep=0.75pt]   [align=left] {$x_{2k-3}$};
		\draw (144,240) node [anchor=north west][inner sep=0.75pt]   [align=left] {$x_{2(k-i)}$b};
		\draw (144,277) node [anchor=north west][inner sep=0.75pt]   [align=left] {$x_{2(k-i)-1}$};
		\draw (144,350) node [anchor=north west][inner sep=0.75pt]   [align=left] {$x_{4}$b};
		\draw (144,395) node [anchor=north west][inner sep=0.75pt]   [align=left] {$x_{3}$};
		\draw (278.52,180) node [anchor=north west][inner sep=0.75pt]   [align=left] {bb};
		\draw (242,213) node [anchor=north west][inner sep=0.75pt]   [align=left] {a};
		\draw (242,265) node [anchor=north west][inner sep=0.75pt]   [align=left] {b};
		\draw (244,221) node [anchor=north west][inner sep=0.75pt]   [align=left] {\vdots};
		\draw (145,186) node [anchor=north west][inner sep=0.75pt]   [align=left] {\vdots};
		
		\draw (165,393) node [anchor=north west][inner sep=0.75pt]   [align=left] {bb};
		\draw (189,375) node [anchor=north west][inner sep=0.75pt]   [align=left] {b};
		\draw (192,325) node [anchor=north west][inner sep=0.75pt]   [align=left] {\vdots};
		\draw (189,360.73) node [anchor=north west][inner sep=0.75pt]   [align=left] {a};
		\draw (189.9,308) node [anchor=north west][inner sep=0.75pt]   [align=left] {a};
		\draw (242,250) node [anchor=north west][inner sep=0.75pt]   [align=left] {a};
		\draw (217.2,277) node [anchor=north west][inner sep=0.75pt]   [align=left] {bb};
		\draw (295,15) node [anchor=north west][inner sep=0.75pt]   [align=left] {a$x_{2k}$bb};
		\draw (498,15) node [anchor=north west][inner sep=0.75pt]   [align=left] {b};
		\draw (498,35) node [anchor=north west][inner sep=0.75pt]   [align=left] {b};
		\draw (498,52) node [anchor=north west][inner sep=0.75pt]   [align=left] {b};
		\draw (498,65) node [anchor=north west][inner sep=0.75pt]   [align=left] {\vdots};
		\draw (498,96.51) node [anchor=north west][inner sep=0.75pt]   [align=left] {b};
		\draw (498,137) node [anchor=north west][inner sep=0.75pt]   [align=left] {b};
		\draw (498,154) node [anchor=north west][inner sep=0.75pt]   [align=left] {a};
		\draw (498,168) node [anchor=north west][inner sep=0.75pt]   [align=left] {a};
		\draw (498,180) node [anchor=north west][inner sep=0.75pt]   [align=left] {b};
		\draw (498,218) node [anchor=north west][inner sep=0.75pt]   [align=left] {a};
		\draw (498,265) node [anchor=north west][inner sep=0.75pt]   [align=left] {a};
		\draw (498,230) node [anchor=north west][inner sep=0.75pt]   [align=left] {\vdots};
		\draw (498,277) node [anchor=north west][inner sep=0.75pt]   [align=left] {b};
		\draw (498,312) node [anchor=north west][inner sep=0.75pt]   [align=left] {a};
		\draw (498,377) node [anchor=north west][inner sep=0.75pt]   [align=left] {a};
		\draw (498,336) node [anchor=north west][inner sep=0.75pt]   [align=left] {\vdots};
		\draw (498,390.64) node [anchor=north west][inner sep=0.75pt]   [align=left] {b};
		\draw (498.71,187) node [anchor=north west][inner sep=0.75pt]   [align=left] {\vdots};
		\draw (56.5,62.25) node   [align=left] {\begin{minipage}[lt]{47.6pt}\setlength\topsep{0pt}
			Top part
			\end{minipage}};
		\draw (52.5,257) node   [align=left] {\begin{minipage}[lt]{68pt}\setlength\topsep{0pt}
			Middle part
			\end{minipage}};
		\draw (50,487) node   [align=left] {\begin{minipage}[lt]{68pt}\setlength\topsep{0pt}
			Bottom part
			\end{minipage}};
		\draw (551.21,56.5) node [anchor=north west][inner sep=0.75pt]   [align=left] {$F_{2k-1} -k+1$};
		\draw (550.71,120) node [anchor=north west][inner sep=0.75pt]   [align=left] {$F_{0}$};
		\draw (550.71,162) node [anchor=north west][inner sep=0.75pt]   [align=left] {$F_{2}$};
		\draw (550.71,187) node [anchor=north west][inner sep=0.75pt]   [align=left] {\vdots};
		\draw (550.71,238) node [anchor=north west][inner sep=0.75pt]   [align=left] {$F_{2i}$};
		\draw (550.71,283 ) node [anchor=north west][inner sep=0.75pt]   [align=left] {\vdots};
		\draw (550.71,342) node [anchor=north west][inner sep=0.75pt]   [align=left] {$F_{2k-4}$};
		\draw (550.71,492) node [anchor=north west][inner sep=0.75pt]   [align=left] {$F_{2k-2}$};
		\draw (145,282) node [anchor=north west][inner sep=0.75pt]   [align=left] {\vdots};
		\draw (498,415) node [anchor=north west][inner sep=0.75pt]   [align=left] {b};
		\draw (498.71,283) node [anchor=north west][inner sep=0.75pt]   [align=left] {\vdots};
		\draw (295,415) node [anchor=north west][inner sep=0.75pt]   [align=left] {ba$x_{2k}$b};
		\draw (498,440) node [anchor=north west][inner sep=0.75pt]   [align=left] {a};
		\draw (498,521) node [anchor=north west][inner sep=0.75pt]   [align=left] {a};
		\draw (498,540) node [anchor=north west][inner sep=0.75pt]   [align=left] {a};
		\draw (498,467) node [anchor=north west][inner sep=0.75pt]   [align=left] {\vdots};
		\draw (144,440) node [anchor=north west][inner sep=0.75pt]   [align=left] {b};
		\draw (144,518) node [anchor=north west][inner sep=0.75pt]   [align=left] {b};
		\draw (144,473) node [anchor=north west][inner sep=0.75pt]   [align=left] {\vdots};
		\draw (295,540) node [anchor=north west][inner sep=0.75pt]   [align=left] {bba$x_{2k}$};
		\end{tikzpicture}

		}
		\caption{A sketch of the BWT-matrix of $v^{\rev}$ where $v$ is a Fibonacci-plus word.\label{fig:stdplus-matrix} }
	\end{figure}
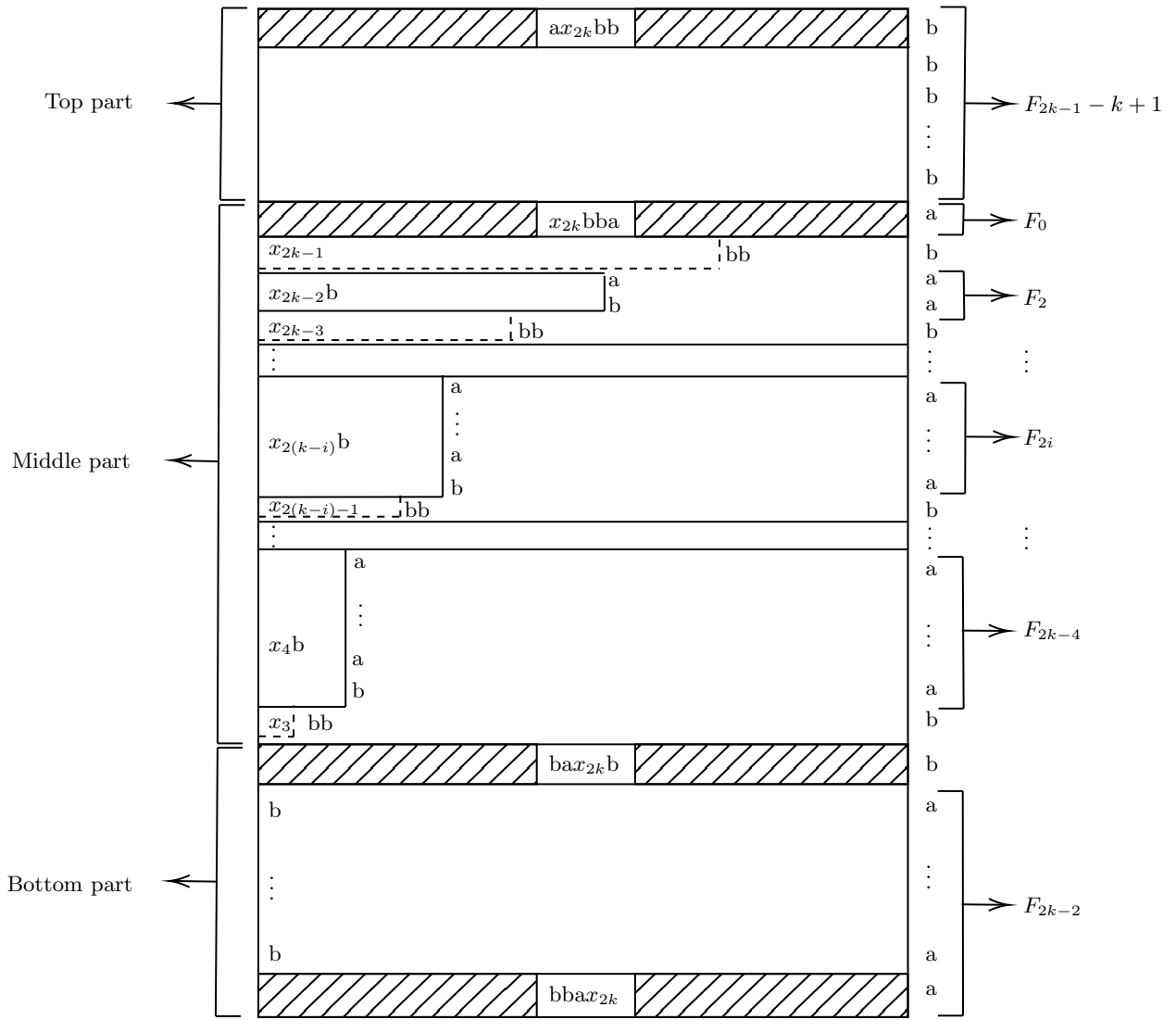


\subsection{Bottom part} 

\begin{proposition}\label{prop:bottom} 
$\bwt(v^{\rev})_{{\rm bot}} = ba^{F_{2k-2}}$. 
\end{proposition} 

\begin{proof} By definition, the bottom part starts with the conjugate $\conj_{2}(v) = bax_{2k}b$. Since $ax_{2k}bb$ is Lyndon (Prop.~\ref{prop:fib}, part 3), it is smaller than all other conjugates, and therefore, $bax_{2k}b$ is smaller than all other conjugates starting with $b$. Thus, the bottom part consists exactly of all conjugates starting with $b$. 
The number of $b$'s in $v$, and thus in $v^{\rev}$ is $F_{2k-2}+1$. Since $s_{2k}$ has no occurrence of $bb$, every $b$ in $v^{\rev}$ except the one in position $2$ is preceded by an $a$, thus $bax_{2k}b$ is the only conjugate ending in $b$. This proves the claim. 
\end{proof}


\subsection{Middle part} 

\begin{lemma}\label{lemma:leftspecial} 
The left-special circular factors of $v^{\rev}$ are exactly the prefixes of $x_{2k-1}b$ and the prefixes of $bax_{2k-2}$. 
\end{lemma}

\begin{proof}
Let $u$ be a left-special circular factor of $v^{\rev}=bbax_{2k}$. From Proposition \ref{prop:fib}, $v^{\rev}=bbax_{2k-1}bax_{2k-2}=bbax_{2k-2}abx_{2k-1}$. Since $bb$ occurs only once, $u$ does not contain $bb$ as factor. Moreover, from combinatorial properties of standard words (see \cite{BorelR06}), it is known that for each $0\leq h\leq F_{2k}-2$, there is exactly one left-special circular factor of $bax_{2k}$ having length $h$ and it a prefix of $x_{2k}$. Since $x_{2k-1}ba$ (that is a prefix of $x_{2k}$) occurs exactly once in $v^{\rev}$ and $bax_{2k-2}$ has exactly two occurrences (one preceded by $b$ and followed by $a$, the other one preceded by $a$ and followed by $b$), either $u$ is prefix of $x_{2k-1}b$ or it is prefix of $bax_{2k-2}$. 
\end{proof}

\begin{lemma}\label{lemma:occsof_x}
Let $s_{2k}$ be a Fibonacci word of even order. 
Then, for all $i=0,\ldots,k-2$, $ax_{2(k-i)}b$ and $ax_{2(k-i)-1}b$ have $F_{2i}$ and $F_{2i+1}$ occurrences, respectively, as circular factors of $s_{2k}$. 
\end{lemma}
\begin{proof}
The statement can be proved by induction on $i$. For $i=0$, the statement follows from the fact that $ax_{2k}b$ and $ax_{2k-1}b$ have just $1=F_0=F_1$ occurrence. Let us suppose the statement is true for all $j\leq i$. Note that $ax_{2(k-i)-2}b$ appears as suffix of $ax_{2(k-i)}b$ and as suffix of $ax_{2(k-i)-1}b$. Moreover, such two occurrences are distinct because $ax_{2(k-i)-1}b$ is not a suffix of $ax_{2(k-i)}b$. This means that, by using the inductive hypothesis, the number of occurrences of $ax_{2(k-i)-2}b$ is $F_{2i}+F_{2i+1}=F_{2i+2}$. Analogously, $ax_{2(k-i)-3}b$ appears as prefix of  $ax_{2(k-i)-1}b$ and as prefix of $ax_{2(k-i)-2}b$. Moreover, such two occurrences are distinct because $ax_{2(k-i)-2}b$ is not a prefix of $ax_{2(k-i)-1}b$. This means that the number of occurrences of $ax_{2(k-i)-3}b$ is $F_{2i+1}+F_{2i+2}=F_{2i+3}$.
\end{proof}

\begin{proposition}\label{prop:mid}
$\bwt(v^{\rev})_{{\rm mid}} = a^{F_0}ba^{F_2}b\ldots a^{F_{2k-4}}b$. 
\end{proposition} 
\begin{proof}
For all $2\leq i<j$, $x_i$ is a prefix (and also a suffix) of $x_j$. This means that the rotations starting with $x_ibb$ are lexicographically greater than $x_jbb$. Moreover, for $1\leq i\leq k-2$,  $x_{2(k-i)}b$ is not a prefix of $x_{2k-1}b$. Thus, by Lemma \ref{lemma:leftspecial}, $x_{2(k-i)}b$ is not left-special. 
Therefore, each occurrence of $x_{2(k-i)}b$ is preceded by the same character; this character must be $a$, since otherwise, both $bx_{2(k-i)}b$ and $ax_{2(k-i)}a$ would be factors, contradicting the fact that $s_{2k}^{\rev}$ is balanced (Prop.~\ref{prop:fib}, part 4). 
Therefore, all occurrences of $x_{2(k-i)}b$ correspond to a run of $a$'s in the $BWT$. The length of this run is $F_{2i}$ by Lemma \ref{lemma:occsof_x}. The claim follows from the fact that each $x_{2(k-i)-1}bb$ occurs exactly once and it is preceded by $b$.
\end{proof}


\subsection{Top part}

\begin{lemma}\label{lemma:u-top}
Let $i$ be such that $\conj_i(v^{\rev}) < x_{2k}bba$. Then the last character of $\conj_i(v^{\rev})$ is $b$. 
\end{lemma}

\begin{proof} 
Let $u = lcp(\conj_i(v^{\rev}), x_{2k}bba)$. Then $u$ is a proper prefix of $x_{2k-1}$. This is because there are only two occurrences of $x_{2k-1}$, one followed by $ba$, this is the prefix of $x_{2k}bba$, and the other followed by $bb$, thus greater than $x_{2k}bba$. Therefore, $u'=ua$ is a prefix of $\conj_i(v^{\rev})$ but not of $x_{2k-1}$, and thus by Lemma~\ref{lemma:leftspecial} it is not left-special. Now assume that $\conj_i(v^{\rev})$ ends with $a$. Then $aua$ is a factor of $v^{\rev}$, and since $u$ does not contain $bb$, it is thus also a factor of $s_{2k}^{\rev}$. On the other hand, $ub$ is left-special, since it is a prefix of $x_{2k-1}b$ (Lemma~\ref{lemma:leftspecial}), therefore both $bub$ and $aua$ are factors of $v^{\rev}$, and again, of $s_{2k}^{\rev}$. This implies that both $au^{\rev}a$ and $bu^{\rev}b$ are factors of $s_{2k}$. This is a contradiction, since $s_{2k}$ is balanced (Prop.~\ref{prop:fib}, part 4). 
\end{proof}

\begin{proposition}\label{prop:top} 
$\bwt(v^{\rev})_{{\rm top}} = b^{F_{2k-2}-k+1}$. 
\end{proposition} 

\begin{proof} 
By Lemma~\ref{lemma:u-top}, $\bwt(v^{\rev})_{{\rm top}}$ consists of $b$'s only. The number of $b$'s of $v$ is $F_{2k-2}+1$, of which we have accounted for $k$ (since $1$ is contained in $\bwt(v^{\rev})_{{\rm bot}}$ and $k-1$ in $\bwt(v^{\rev})_{{\rm mid}}$), there remaining exactly ${F_{2k-2}-k+1}$ $b$'s. 
\end{proof}


\subsection{Putting it all together}

\begin{proof} {\em of Prop.~\ref{prop:fibplus-r}}: The claim for even-order Fibonacci-plus words follows from Propositions~\ref{prop:bottom},~\ref{prop:mid}, and~\ref{prop:top}. The claim for odd-order Fibonacci-plus words can be proved in an analogous manner. 
\end{proof}

\begin{proof} {\em of Thm.~\ref{thm:fibplus}}: From Propositions~\ref{prop:fibplus-r} and~\ref{prop:fibplusrev}, we have that $\rho(v) = 2k/4 = k/2$.  On the other hand, $n = |v| = F_{2k}+1$, thus by the properties of the Fibonacci numbers, $2k = \Theta(\log n)$, implying that $\rho(v) = k/2 = \Theta(\log n)$. \end{proof}


\section{Standard-plus words have $\rho = \Oh(\log n)$} 
In this section we consider other infinite families of finite words, defined from standard words. Here we assume that $d_0\geq 1$, otherwise we could consider the word obtained by exchanging $a$'s and $b$'s and the results still hold true.

\begin{definition}
A word $v$ is called {\em standard-plus} if it is either of the form $sb$, where $s$ is a standard word of even order $2k$, $k\geq 2$,  or of the form $sa$, where $s$ is a standard word of odd order $2k+1$,  $k\geq 2$. In the first case, $v$ is {\em of even order}, otherwise {\em of odd order}.  
\end{definition}

We show that, when a standard-plus word $v=s_{2k}b$ is considered, the exact asymptotic growth of $\rho$ depends on the directive sequence of the word $s_{2k}$. Here we give the proof of the result for standard-plus words of even order, however an analogous statement can also be proved for standard-plus words of odd order.

\begin{proposition}\label{prop:stdplus-r} Let $v=s_{2k}b$ be a standard-plus word of even order. Then $r(v) = 4$.
\end{proposition}

The proof of Proposition \ref{prop:stdplus-r} is analogous to that of Proposition \ref{prop:fibplus-r}.

\begin{proposition}\label{prop:rho_stdplus_rev}
Let $v=s_{2k}b$ be a standard-plus word of even order $2k$, where $s_{2k}$ is the standard word obtained by using the directive sequence $(d_0,d_1,\ldots,d_{2k-2})$ of length $2k-1$, where $d_0\geq 1$. If $d_0=1$, then $r(v^{\rev})=2k$. Otherwise,  $r(v^{\rev})=2k+2$.

\end{proposition}
\begin{proof}
(Sketch) Similar to what happens with Fibonacci's words (see Prop. \ref{prop:fib}), it is known that $s_{2k}=Cab$, where $C$ is a palindrome, the conjugate $aCb$ is a Lyndon word (see \cite{deLucaMignosi1994,BersteldeLuca97}). Then  $v^{\rev}=bbaC$ and, in order to lexicographically sort the conjugates of $v^{\rev}$, we can consider its Lyndon rotation $aCbb$. 
One can verify that $C\in\{a^{d_0}b,a^{d_0+1}b\}^*$.
It is possible to see that $\bwt(v^{\rev})$ ends with $ba^{|s_{2k}|_b}$, since $baCb$ is the smallest rotation starting with $b$. Moreover, since $t=b(a^{d_0}b)^{d_1}b$ is a suffix of $aCbb$, all rotations of $v^{\rev}$ starting with the first occurrence of $a$ in each run $a^{d_0}$ in $t$  determine $d_1$ consecutive $b$'s in $\bwt(v^{\rev})$. If $d_0=1$ such rotations are followed by the rotation $baCb$, otherwise several rotations preceded by $a$ (including the rotations starting with the other $a$'s of $t$) are in between. So, if $d_0=1$, the last run of $b$'s has length $d_1+1$, otherwise the last two runs of $b$'s have length $d_1$ and $1$, respectively.

Finally, when $d_i$ (with odd $i$) is used to generate standard words, a set of consecutive rotations starting with $(a^{d_0}b)^{d_1}a^{d_0+1}b$ and preceded by $b$ is produced. This means that the other runs of $b$'s have length $d_3,d_5,\ldots,d_{2k-3},|s_{2k}|_b-(d_1+d_3+\ldots+d_{2k-3})$.
\end{proof}

\begin{example}
Let us consider the standard-plus word $v$ of even order constructed by using the directive sequence $(2,3,1,2,1)$. 
One can verify that $$v=aabaabaabaaabaabaabaabaaabaabaabaabaaabaabaabaaabb.$$
Moreover, $\bwt(v^{\rev})=b^{10}ab^2a^3b^3a^{15}ba^{15}$ and $\bwt(v)=b^{15}a^{33}ba.$
\end{example}

\begin{theorem}
Let $v$ be a standard-plus word of even order $n$. Then $\rho(v)=\Oh(\log n)$. 
\end{theorem}
\begin{proof}
    By definition, $v=s_{2k}b$ where $s_{2k}$ is a standard word of order $n=2k$ for some positive $k$. Since $|s_{2k}|\geq F_{2k}$, by Prop. \ref{prop:stdplus-r} and \ref{prop:rho_stdplus_rev}, $\rho(v)\leq \frac{k+1}{2}\in \Oh(\log n)$.
\end{proof}

The following proposition states that among all standard-plus words, Fibonacci-plus words are maximal w.r.t.\ $\rho$.
\begin{proposition}\label{prop:fibP_highestRho}
Let $v$ be a Fibonacci-plus word, and $v'$ a standard-plus word s.t.\ $|v| = |v'|$. Then $\rho(v) \geq \rho(v')$. 	
\end{proposition}

\begin{proof}
	Follows directly from Prop.~\ref{prop:stdplus-r} and~\ref{prop:rho_stdplus_rev}, and from the fact that  Fibonacci words have the longest directive sequence among all standard words of the same length.
\end{proof}


\section{Conclusion and Outlook} 

In this paper, we presented the first non-trivial lower bound on the maximum runs-ratio $\rho(n)$ of a word of length $n$. This shows for the first time that the widely used parameter $r$, the number of runs of the BWT of a word, is not an ideal measure of the repetitiveness of the word. Moreover, it proves that for BWT-based compression a parallel result holds to the ``one-bit catastrophe'' recently shown for LZ78-compression~\cite{Lagarde18}. 

\medskip

Several open questions remain. We saw in the previous section that Fibonacci-plus words are maximal among the class of standard-plus words with respect to the runs-ratio $\rho$. However, they stay strictly below $\rho(n)$, the maximum among all words of length $n$,  even for lengths up to $n=30$. In Table~\ref{tab:rhon2}, we report the values of $\rho(n)$ and compare them to the maximum reached by standard-plus words. Note that this is a Fibonacci-plus word only for $n=9,14,22$. 

	\begin{table}\centering
		\begin{tabular}{c|*{24}{c}}
			$n$ & 
			9 & 10 & 11 & 12 & 13 & 14 & 15 & 16 & 17 & 18 & 19 & 20 & 21 & 22 & 23 & 24 & 25 & 26 & 27 & 28 & 29 & 30\\
			\hline
			$\rho(n)$ & 
			2 & 2 & 2 & 2 & 2 & 2 & 2 &  2 & 2.5 & 2.5 & 2.5 & 2.5 & 3 & 2.5 & 3 & 3 & 2.67 & 3 & 3 & 3 & 3 & 3\\
			
			std-plus & 
			1 & 1.5 & 1.5 & 1.5 & 1 & 1.5 & 1.5 & 1.5 & 1.5 & 1.5 & 2 & 1.5 & 1.5 & 1.5 & 1.5 & 2 & 1.5 & 2 & 2 & 2 & 2 & 2 \\
		\end{tabular}
		\caption{The values of $\rho(n)$ for $n=9,\ldots, 30$, and the maximum value of $\rho(n)$ among all standard-plus words of length $n$. \label{tab:rhon2}}
	\end{table}

It is possible to construct binary words of arbitrary length and greater runs-ratio $\rho$ than any standard-plus word of the same length. However, we currently do not know the asymptotic growth of the $\rho$ value for such words. Therefore, the question of closing the gap for $\rho(n)$ between our lower bound $\Omega(\log n)$ and the upper bound $\Oh(\log^2(n))$ remains open. 

It would be interesting to explore the question also for larger alphabets. Our preliminary experimental results on ternary alphabets indicate that the increase in $\rho$ happens at smaller lengths than for the binary case. This suggests that the effect we showed in this paper, of a divergence between the string's repetitiveness and $r$, may be even more pronounced in real-life applications.


\subsubsection*{Acknowledgements}
Zs.L.\ and M.S.\ wish to thank Dominik Kempa for getting them interested in the problem treated in this paper. We thank Gabriele Fici and Daniele Greco for interesting discussions, and Akihiro Nishi for preliminary experiments. We thank the Leibniz Zentrum f\"ur Informatik for the possibility of participating at  Dagstuhl Seminar no.\ 19241 in June 2019, where some of the authors started collaborating on this problem. 


\end{document}